\documentclass[a4paper]{amsart}

\usepackage{amsmath}
\usepackage{amssymb}
\usepackage{amsthm}
\usepackage{graphicx}
\usepackage{bbm}

\newtheorem{theorem}{Theorem}

\newtheorem{lemma}[theorem]{Lemma}

\newtheorem{corollary}[theorem]{Corollary}

\newcommand{\expec}{\mathbb{E}}

\newcommand{\ind}{\mathbf{1}}

\title[Asymptotics for the Merton and Kou jump diffusion models]{Refined wing asymptotics for the Merton and Kou jump diffusion models}

\author{Stefan Gerhold}
\address{Vienna University of Technology, Wiedner Hauptstra\ss{}e 8--10/105-1,
A-1040 Vienna, Austria}
\email{sgerhold@fam.tuwien.ac.at}
\thanks{This note is based in part on the thesis~\cite{Zr13},
where some proofs are discussed in greater detail.
S.~Gerhold gratefully acknowledges financial support from the
Austrian Science Fund (FWF) under grant P~24880-N25.}

\author{Johannes F. Morgenbesser}
\address{University of Vienna, Oskar-Morgenstern-Platz 1, A-1090 Vienna, Austria}
\email{johannes.morgenbesser@univie.ac.at}

\author{Axel Zrunek}
\email{axel.zrunek@aon.at}

\date{\today}

\begin{document}

\begin{abstract}
  Refining previously known estimates, we give large-strike asymptotics
  for the implied volatility of Merton's and Kou's jump diffusion models.
  They are deduced from call price approximations by transfer
  results of Gao and Lee. For the Merton model, we also analyse the density of the underlying and show that it features an interesting ``almost power law'' tail.  
\end{abstract}



\maketitle

\section{Introduction}

In recent years, the literature on asymptotic approximations of 
option prices and implied volatilities has grown at a fast pace.
Important papers on large-strike asymptotics include
\cite{BeFr08,BeFr09,Gu10,Le04a}.
Of particular relevance to the present note is the approach
of Gao and Lee~\cite{GaLe11}, who translate call price
asymptotics to implied volatility asymptotics, robustly
w.r.t.\ choice of model and asymptotic regime.
So far, relatively few models have been analyzed in sufficient
detail to use the full power of their transfer results.
We extend concrete
applicability of some results of~\cite{GaLe11} by presenting refined
strike asymptotics
for the well-known jump diffusion models of Merton~\cite{Me76} and Kou~\cite{Ko02}.
Potential practical consequences of our expansions concern
fast calibration and implied volatility parametrization resp.\ extrapolation.

As we are in a fixed-maturity regime, we can set the interest rate
to zero throughout. Also, initial spot is normalized to $S_0=1$.
Log-returns are modeled by a L\'evy jump diffusion 
\[
  X_t = b t + \sigma W_t + \sum_{j=1}^{N_t} Y_j
\]
with drift $b\in\mathbb{R}$ and
diffusion volatility $\sigma>0$. The process~$W$ is  a standard Brownian motion,
$N$ is a Poisson process with intensity $\lambda>0$, and the jumps
$Y_j$ are i.i.d. real random variables. As for the law of the~$Y_j$,
we focus on the double exponential (Kou) and Gaussian (Merton) cases.
The (dimensionless) implied volatility $V(k)$ is the solution
of
\[
  c_{\,\mathrm{BS}}(k,V(k)) = \expec[(S_T - e^k)^+],
\]
where
\[
  c_{\,\mathrm{BS}}(k,v) = \Phi(d_1) - e^k \Phi(d_2)
\]
is the Black-Scholes call price,
with $d_{1,2}=-k/v \pm v/2$ and $\Phi$ the standard Gaussian cdf.
Our interest is in the growth order of~$V(k)$ as $k\to\infty$.
While first order asymptotics are known for both models we treat,
they suffer from limited practical applicability. Higher
order terms typically increase accuracy significantly, even for
moderate values of the log-strike~$k$.

We observe that the large-strike behavior of the Kou model
is of the same shape (in terms of the asymptotic elements
involved) as for the Heston model~\cite{FrGeGuSt11}. Not obvious from
the respective model dynamics, this fact is an immediate consequence
of the local behavior of the moment generating function (mgf) at the
critical moments. This behavior was analyzed in~\cite{FrGeGuSt11}
from affine principles, whereas the present analysis for the Kou
model profits from its very simple explicit mgf.

For the Merton model, we include an approximation of the density
(Theorem~\ref{mjd_density}). The reason is that it implies
an interesting ``almost power law'' tail for the marginals of the underlying,
of order $k^{-\sqrt{\log \log k}}$.

Throughout the paper, $F(k) \ll G(k)$ means that the functions~$F$
and~$G$ satisfy $F(k) = O(G(k))$ for $k\to\infty$.

\section{Kou jump diffusion model}

In the Kou model~\cite{Ko02},
the~$Y_j$ are double exponentially distributed, and thus
have the common density
\[
f(y) = p \lambda_+e^{-\lambda_+y}\ind_{[0,\infty)}(y) + (1-p)\lambda_-e^{\lambda_- y}\ind_{(-\infty,0)}(y)
\]
with parameters $\lambda_+ >1$, $\lambda_- > 0$ and $p\in (0,1)$.
One of the advantages of this model is the memoryless property
of the double exponential distribution, which leads to analytical
formulas for several types of options.
The moment generating function of the log-price~$X_T$ is given by
\begin{align}
M(s,T) &= \expec[\exp(sX_T)] \notag \\
&= \exp{\left(T\left(\frac{\sigma^2s^2}{2}+bs+\lambda\left(\frac{\lambda_+p}{\lambda_+-s}+\frac{\lambda_-(1-p)}{\lambda_-+s}-1\right)\right)\right)}. \label{kou_M}
\end{align}

From Benaim and Friz' refinement of Lee's moment formula (Example~5.3
in~\cite{BeFr08}), it is known that~$V$  has the first order asymptotics
\begin{equation}\label{kou_firstorder}
\lim_{k \to \infty} \frac{V(k)}{k^{1/2}} = \Psi^{1/2} (\lambda_+ -1),
\end{equation}
where $\Psi(x)$ is defined by 
\begin{equation}\label{psi}
\Psi(x) = 2 -4(\sqrt{x^2+x}-x).
\end{equation}
To formulate our refined expansion for the Kou call price, define
\begin{align*}
 \alpha_1 &= \lambda_+ -1, \qquad \alpha_{1/2} = -2(\lambda\lambda_+pT)^{1/2}, \\
 \alpha_0 &=
T\left(-\frac{\sigma^2\lambda_+^2}{2}-b\lambda_+-\frac{\lambda\lambda_-(1-p)}{\lambda_-+\lambda_+}+\lambda \right) 
 -\log\frac{(\lambda\lambda_+pT)^{1/4}}{2\sqrt{\pi}\lambda_+(\lambda_+-1)}.
\end{align*}
For better readability, the notation here ($\alpha_i$, and~$\beta_i$ below)
is as in Corollary~7.11 of~\cite{GaLe11}; coefficient indices
mimic the asymptotic terms they belong to.
\begin{theorem}\label{thm:call kou}
The price of a call option in the Kou model satisfies
\begin{equation}\label{eq:kou call}
  C(k,T) =    
    \exp\left(-\alpha_1k-\alpha_{1/2}k^{-1/2}
    -\alpha_0\right)k^{-3/4}(1+O(k^{-1/4}))
\end{equation}
as $k \to \infty$.
\end{theorem}
The call price expansion~\eqref{eq:kou call} is amenable to the general
transfer results of Gao and Lee~\cite{GaLe11}. Indeed, their Corollary~7.11
immediately implies the following implied volatility expansion.
\begin{corollary}\label{kou_impvola_asym}
The implied volatiliy of the Kou model satisfies, as $k\to\infty$,
\begin{equation}\label{eq:kou iv}
V(k) = \beta_{1/2}k^{1/2} + \beta_0 + \beta_{\ell -1/2}\frac{\log k}{k^{1/2}} + \beta_{-1/2}\frac{1}{k^{1/2}} + O(k^{-3/4}),
\end{equation}
where
\begin{align*}
  \beta_{1/2} &= -2\gamma\sqrt{\alpha_1^2+\alpha_1} = \Psi^{1/2}(\lambda_+-1), \\   
  \beta_0 &= \gamma\alpha_{1/2}, \qquad \beta_{\ell-1/2}=\frac{\gamma}{4}, \\
  \beta_{-1/2} &= \left(\alpha_0 + \log\frac{1-(1+\frac{1}{\alpha_1})^{-1/2}}{\sqrt{4\pi\alpha_1}}\right)\gamma + \left(\frac{1}{2(2\alpha_1)^{3/2}}-\frac{1}{2(2\alpha_1+2)^{3/2}}\right)\alpha_{1/2}^2,
\end{align*}
and
\[
\gamma = \left(\frac{1}{\sqrt{2\alpha_1+2}}-\frac{1}{\sqrt{2\alpha_1}}\right).
\]
\end{corollary}
The gain in numerical precision over the first order
approximation~\eqref{kou_firstorder} depends on the model parameters.
See Figures~\ref{fig:kou iv 1} and~\ref{fig:kou iv 2} for examples.
It is an interesting fact that the expansion~\eqref{eq:kou iv} has
the same shape as the implied vol expansion of the Heston
model (see Theorem~1.3 in~\cite{FrGeGuSt11}). While this is not obvious
from the model specifications, it is clear from an asymptotic analysis
of the respective mgfs: Call price and density asymptotics are governed
by the type of singularity found at the critical moment, and this
is ``exponential of a first order pole'' in both cases.
(For other papers analyzing functions with such behavior, see, e.g.,
\cite{FlGeSa10,Ge11b}.)
\begin{figure}
\begin{center}
\includegraphics[height=2.6in]{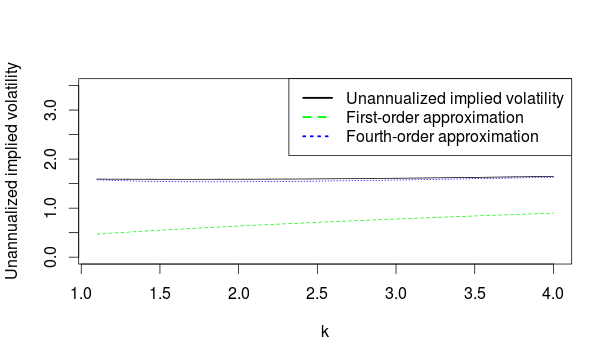}
\end{center}
\caption{\label{fig:kou iv 1}Implied volatility of the Kou model
(solid) compared with the fourth-order expansion~\eqref{eq:kou iv}
(dotted) and Lee's formula~\eqref{kou_firstorder} (dashed).
The parameters are $T=6$, $\sigma=0.4$, $\lambda=1$, $p=0.2$, $\lambda_-=2$, and $\lambda_+=3$. }
\end{figure}
\begin{figure}
\begin{center}
\includegraphics[height=2.6in]{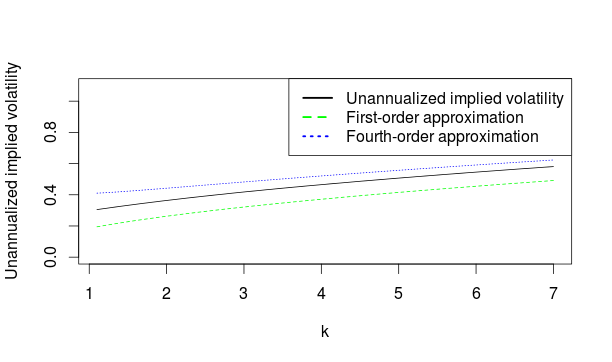}
\end{center}
\caption{\label{fig:kou iv 2} Implied volatility of the Kou model
(solid) compared with the fourth-order expansion~\eqref{eq:kou iv}
(dotted) and Lee's formula~\eqref{kou_firstorder} (dashed). The parameters are $T=1$, $\sigma=0.1$, $\lambda=5$, $\lambda_+=15$, $\lambda_-=15$, $p=0.5$. }
\end{figure}

We briefly comment on the qualitative implications of
Corollary~\ref{kou_impvola_asym}. The dominating term depends only
on~$\lambda_+$, i.e., the parameter that governs the propensity
to jump in the in-the-money direction. If~$\lambda_+$ increases,
upward jumps become smaller, and
the call clearly becomes cheaper, thus lowering implied volatility.
(Note that the function~$\Psi$ is decreasing.)
Second order asymptotics of implied vol, i.e., $\beta_0$
in~\eqref{eq:kou iv}, depend additionally on $\lambda$, $p$, and $T$.
It is remarkable that the smile wings are very insensitive to the
diffusion volatility~$\sigma$ and the downwards-jump parameter~$\lambda_-$,
which appear only in the $k^{-1/2}$-term in~\eqref{eq:kou iv}.

As the transfer from call price to implied vol asymptotics is
handled by~\cite{GaLe11} in a mechanical way, the rest of this section
is devoted to the proof of Theorem~\ref{thm:call kou}.
The mgf~\eqref{kou_M} is analytic in the strip $\Re s \in
(\lambda_-,\lambda_+)$.
By the exponential decay of $M(s,T)$ for $|\Im s| \to\infty$, the Fourier
representation
\begin{equation}\label{call_mellin}
C(k,T) = \frac{e^k}{2\pi i} \int_{\eta-i\infty}^{\eta+i\infty} e^{-ks}\frac{M(s,T)}{s(s-1)}ds
\end{equation} 
of the call price is valid; see Lee~\cite{Le04b}. The real 
part of the integration contour satisfies $1<\eta<\lambda_+$.
We prove Theorem~\ref{thm:call kou} by a saddle point analysis
of~\eqref{call_mellin}. 
(The series representation of the call price from Kou's paper~\cite{Ko02}
seems less amenable to asymptotic analysis.)
The integrand blows up as $s\to\lambda_+$.
Identifying the dominating term $1/(\lambda_+-s)$ in~\eqref{kou_M},
we define the (approximate) saddle point $\hat{s}=\hat{s}(k)$
as the solution of
\begin{equation*}
  \frac{\partial}{\partial s}e^{-ks}
  \exp\left(T\lambda\frac{\lambda_+p}{\lambda_+-s}\right) = 0,
\end{equation*}
which is given by
\begin{equation*}%\label{kou_saddlepoint}
  \hat{s} =  \lambda_+-\xi^{1/2}k^{-1/2},
\end{equation*}
where $\xi = \lambda\lambda_+pT$.
\begin{lemma}\label{kou_asym}
The cumulant generating function $m(s,T)=\log M(s,T)$ of $X_T$ satisfies 
\begin{align*}
 m(\hat{s},T) &=\frac{T\sigma^2\lambda_+^2}{2}+ b\lambda_+ T+ \xi^{1/2}k^{1/2} +  \frac{T\lambda\lambda_-(1-p)}{\lambda_-+\lambda_+} - \lambda T+ O(k^{-1/2}),  \\
m'(\hat{s},T) &= k + O(1), \\
m''(\hat{s},T) &= 2\xi^{-1/2}k^{3/2} + O(1), \\
 \quad m'''(\hat{s}+it,T) &= O(k^2) \quad \text{for } |t|<k^{-\alpha},~~\alpha >0,
\end{align*}
where all derivatives are with respect to $s$.
\end{lemma}
\begin{proof}
  The estimates follow by straightforward computations from~\eqref{kou_M}.
\end{proof}
We now move the integration contour so that it passes through
the saddle point~$\hat{s}$. A small part of the contour, within
distance $O(k^{-\alpha})$
of the saddle point, captures the asymptotics of the full integral.
Any exponent $\alpha \in (\tfrac23, \tfrac34)$ is suitable. The integration
variable thus becomes $s=\hat{s}+i t$, $|t|<k^{-\alpha}$.
From the estimates for $m'$ and $m'''$ in
Lemma~\ref{kou_asym}, we have the local expansion
\[
  M(\hat{s}+it,T) = \exp\left(m(\hat{s},T) 
  + itk-\frac{m''(\hat{s},T)}{2}t^2\right)(1+O(t+t^3k^2)).
\]
The rational function $1/(s(s-1))$ is locally constant, to first order:
\[
  \frac{1}{(\hat{s}+it)(\hat{s}+it-1)} = 
  \frac{1}{\lambda_+(\lambda_+-1)}(1+O(k^{-1/2})).
\]
The integral close to the saddle point thus becomes
\begin{multline}\label{eq:int central}
  \frac{e^k}{2\pi i}\int_{\hat{s}-ik^{-\alpha}}^{\hat{s}+ik^{-\alpha}}
    e^{-ks}\frac{M(s,T)}{s(s-1)}ds \\
=\frac{e^{k(1-\hat{s})}M(\hat{s},T)}{2\pi\lambda_+(\lambda_+-1)}\int_{-k^{-\alpha}}^{k^{-\alpha}}\exp\left(-\frac{m''(\hat{s},T)}{2}t^2\right)(1+O(k^{-3\alpha+2}))dt.
\end{multline}
Setting $u:=m''(\hat{s},T)^{1/2}$, we get from Lemma~\ref{kou_asym}
\[
 u =  \frac{\sqrt{2}k^{3/4}}{\xi^{1/4}}(1+O(k^{-3/2})) \quad
 \text{and} \quad
  \frac{1}{u} =\frac{\xi^{1/4}}{\sqrt{2}k^{3/4}}(1+O(k^{-3/2})).
\]
By substituting $\omega = ut$ and using the fact that Gaussian integrals have exponentially decaying tails, we find
\begin{align}
\int_{-k^{-\alpha}}^{k^{-\alpha}}\exp\left(-\frac{m''(\hat{s},T)}{2}t^2\right)dt
  &=
\frac{1}{u}\int_{-uk^{-\alpha}}^{uk^{-\alpha}}\exp\left(-\frac{\omega ^2}{2}\right)d\omega  \notag \\
&= \sqrt{\pi}k^{-3/4}\xi^{1/4}(1 + O(k^{-3/2})). \label{eq:gauss}
\end{align}
Use this in~\eqref{eq:int central}, and replace $M(\hat{s},T)$ by
the estimates in Lemma~\ref{kou_asym}, to get the right hand side
of~\eqref{eq:kou call}.
Note that to obtain a relative error $k^{-1/4}$, and not just
$k^{-1/4+\varepsilon}$, it suffices to take one further term in the
local expansion, keeping in mind the well-known fact that the saddle
point method usually yields a full asymptotic expansion.
See~\cite{FrGeGuSt11} for a detailed discussion of this issue in
a similar analysis.
(The same remark applies in case of the Merton model below.)
To prove Theorem~\ref{thm:call kou}, it thus remains
to show that the integral over $|t|> k^{-\alpha}$ can be dropped, so
that~\eqref{eq:int central} asymptotically equals~\eqref{call_mellin}.
This tail estimate is done in the following lemma. By symmetry, it suffices
to treat $t>k^{-\alpha}$.

\begin{lemma} Let $\tfrac23 < \alpha < \tfrac34$. Then we have
\[
\frac{e^k}{2\pi i} \int_{\hat{s}+ik^{\alpha}}^{\hat{s} +i \infty}e^{-ks}\frac{M(s,T)}{s(s-1)}ds \ll e^{k(1-\lambda_+)+2\xi^{1/2}k^{1/2}-\xi^{-1/2}k^{3/2-2\alpha}/2}.
\]
\end{lemma}
\begin{proof}
Let $s=\hat{s}+it = \lambda_+-\xi^{1/2}k^{-1/2}+it$, where $t\geq k^{-\alpha}$.
Then
\begin{align*}
|M(s,T)| &\ll \exp\left(\Re\left(\frac{\xi}{\xi^{1/2}k^{-1/2}-it}\right)\right) \\
&= \exp\left(\frac{\xi^{1/2}k^{1/2}}{1+t^2\xi^{-1}k}\right) \ll \exp\left(\frac{\xi^{1/2}k^{1/2}}{1+k^{1-2\alpha}\xi^{-1}}\right).
\end{align*}
Using the fact $1/(1+x) \leq 1-x/2$ for $x \leq 1$ and that $k^{1-2\alpha}$ is smaller than 1 for sufficiently large $k$, we get
\begin{equation}\label{kou_tails2}
|M(s,T)| \ll \exp\left(\xi^{1/2}k^{1/2}-\xi^{-1/2}k^{3/2-2\alpha}/2\right).
\end{equation}
From $|s(s-1)| \gg 1+t^2$ and~\eqref{kou_tails2}, we obtain
\[
\frac{e^k}{2\pi}\left|e^{-ks}\frac{M(s,T)}{s(s-1)}\right| \ll \frac{e^{k(1-\lambda_+)+2\xi^{1/2}k^{1/2}-\xi^{-1/2}k^{3/2-2\alpha}/2}}{1+t^2},
\]
and thus
\begin{align*}
\frac{e^k}{2\pi i}\int_{\hat{s}+ik^{-\alpha}}^{\hat{s}+i\infty} e^{-ks}\frac{M(s,T)}{t(s-1)}ds &\ll e^{k(1-\lambda_+)+2\xi^{1/2}k^{1/2}-\xi^{-1/2}k^{3/2-2\alpha}/2}\int_{k^{-\alpha}}^{\infty} \frac{dt}{1+t^2} \\
 &\ll e^{k(1-\lambda_+)+2\xi^{1/2}k^{1/2}-\xi^{-1/2}k^{3/2-2\alpha}/2}.
\end{align*}
\end{proof}
The proof of Theorem~\ref{thm:call kou} is complete.
Finally, we mention that tail asymptotics for the distribution of~$X_T$
(for $\sigma^2=0$) can be found in~\cite{AlSu09}, Example~7.5.
They can also be deduced from earlier work of Embrechts et al.~\cite{EmJeMaTe85}.

\section{Merton jump diffusion model}

In one of his classical papers, Merton~\cite{Me76} proposed a L\'evy
jump diffusion with Gaussian jumps as a model for log-returns. If mean and variance
of the jump size distribution are~$\mu$ resp.~$\delta^2$,
then the mgf is the entire function
\[
M(s,T) =  \exp\left(T\left(\frac{1}{2}\sigma^2 s^2 + bs + \lambda(e^{\delta^2 s^2/2 + \mu s}-1)\right)\right).
\]
Benaim and Friz~\cite{BeFr09} gave first order logarithmic
asymptotics for the call price,
\begin{equation}\label{eq:merton call 1st order}
  L := -\log C(k,T) \sim \frac{\sqrt{2}}{\delta} k \sqrt{\log k},
    \quad k\to\infty,
\end{equation}
and first order asymptotics for implied volatility:
\begin{equation}\label{eq:merton iv 1st order}
  V(k) \sim 2^{-3/4} \sqrt{\delta k}/(\log k)^{1/4}.
\end{equation}
Our refined results are best formulated using the implicitly defined
saddle point~$\hat{s}=\hat{s}(k)$, satisfying $m'(\hat{s},T)=k$.
(As in the Kou model, we write $m=\log M$ for the cumulant
generating function.)
\begin{theorem}\label{mjd_call_expansion}
For $k \to \infty$, the call price in the Merton model satisfies
\begin{align}
C(k,T) &= \frac{\delta^2 e^{k(1-\hat{s})}M(\hat{s},T)}{2\log k \sqrt{2\pi m''(\hat{s},T)}}\left(1+O\left(\frac{1}{\sqrt{\log k}}\right)\right) 
  \label{eq:merton call s} \\
&= \frac{\delta^2 e^{k(1-\hat{s}) + T\left(\sigma^2\hat{s}^2/2 + b\hat{s} + \lambda\left(e^{\delta^2\hat{s}^2/2 + \mu\hat{s}}-1\right)\right)}}{2\log k\sqrt{2\pi T\left(\sigma^2 + \lambda(\delta^2\hat{s}+\mu)\left((\delta^2\hat{s}+\mu)+\delta^2\right)e^{\delta^2\hat{s}^2/2+\mu\hat{s}}\right)}} \label{eq:merton call} \\
& \qquad  \cdot \left(1+O\left(\frac{1}{\sqrt{\log k}}\right)\right).  \notag
\end{align}
\end{theorem}
\begin{corollary}\label{mjd_impvola_expansion}
The implied volatility of the Merton model satisfies, for $k \to \infty$,
\begin{align}
V(k) &= G_-\left(k,L-\frac{3}{2}\log L + \log \frac{k}{4\sqrt{\pi}} \right)
 + O(k^{-1/2}(\log k)^{-5/4}) \label{eq:merton V} \\
 &= 2^{-3/4} \sqrt{\delta k} (\log k)^{-1/4} +
 c  \sqrt{k} (\log k)^{-3/4} 
  + O\left(\frac{\sqrt{k} \log\log k}{(\log k)^{5/4}} \right),
  \label{eq:merton V expl}
\end{align}
where $L = -\log C(k,T)$ is the absolute log of the call price,
$G_-(k,u)=\sqrt{2}(\sqrt{u+k}-\sqrt{u})$, and
  $c = 2^{-9/4} \delta^{-1/2}(\mu+\delta^2) - 2^{-13/4}\delta^{3/2}$.
\end{corollary}
\begin{theorem}\label{mjd_density} 
The density of the Merton log-return $X_T$ satisfies, for $k \to \infty$,
\begin{align}
f_{X_T}(k) &= \frac{e^{-k\hat{s}}M(\hat{s},T)}{\sqrt{2\pi m''(\hat{s},T)}}(1+O(k^{-1/2})) \label{eq:merton dens raw} \\
&= \frac{e^{-k\hat{s} + T\left(\sigma^2/2\hat{s}^2 + b\hat{s} + \lambda\left(e^{\delta^2/2\hat{s}^2 + \mu\hat{s}}-1\right)\right)}}{\sqrt{2\pi T\left(\sigma^2 + \lambda(\delta^2\hat{s}+\mu)\left((\delta^2\hat{s}+\mu)+\delta^2\right)e^{\delta^2\hat{s}^2/2+\mu\hat{s}}\right)}}(1+O(k^{-1/2}))
  \label{eq:merton dens} \\
   &= \exp\left( -\frac{\sqrt{2}}{\delta} k \sqrt{\log k} + O(k)\right).
 \label{eq:merton dens expl}
\end{align}
\end{theorem}
\begin{figure}
\begin{center}
\includegraphics[height=2.6in]{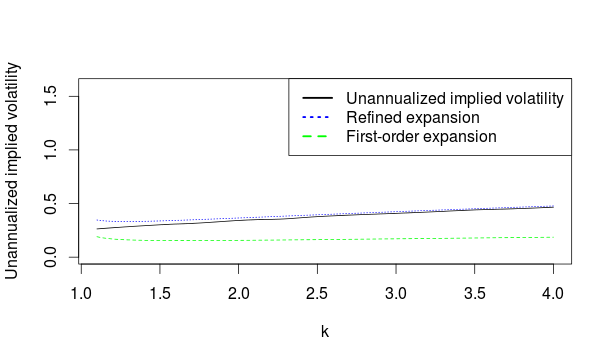}
\end{center}
\caption{\label{fig:merton iv} The solid curve is the implied vol
of the Merton model, with parameters $T=0.1$, $\sigma=0.4$, $\lambda=0.1$,
$\mu=0.3$, and $\delta=0.4$. The dashed line is the first order
approximation~\eqref{eq:merton iv 1st order},
whereas the dotted line is our refined
approximation~\eqref{eq:merton V}.}
\end{figure}
Formula~\eqref{eq:merton call} is just~\eqref{eq:merton call s} with~$M$
replaced by its explicit form, and the same holds for~\eqref{eq:merton dens}
and~\eqref{eq:merton dens raw}.
For numerical accuracy, it is preferable to use~\eqref{eq:merton V}
and~\eqref{eq:merton dens},
and not the more explicit variants~\eqref{eq:merton V expl}
and~\eqref{eq:merton dens expl}.

We begin the proofs by deducing~\eqref{eq:merton V expl}
and~\eqref{eq:merton dens expl} from the semi-explicit
formulas~\eqref{eq:merton V} and~\eqref{eq:merton dens},
using asymptotic approximations of~$\hat{s}$. The saddle point satisfies
%  \frac{\delta^2}{2}\hat{s}^2 = \log k -\frac{\mu}{\delta}\sqrt{2\log k} - %\log\sqrt{\log k} + \frac{\mu^2}{\delta^2}-\log\frac{\sqrt{2}}{\delta} + %\lan{\frac{\log\log k}{\log k}}
%\]
%and
\begin{equation}\label{eq:merton sp}
 \hat{s} = \frac{\sqrt{2\log k}}{\delta} - \frac{\mu}{\delta^2} +O\left(\frac{\log \log k}{\sqrt{\log k}}\right).
\end{equation}
This estimate follows from the saddle point equation  $m'(\hat{s},T)=k$ by
a tedious, but straightforward application of the classical
``bootstrapping'' technique (see, e.g., chapter~22 of~\cite{GrGrLo95}).
From the saddle point equation, we know that
\begin{equation}\label{eq:merton sp2}
  e^{\delta^2\hat{s}^2/2 + \mu \hat{s}} =
  \frac{k/T - \sigma^2\hat{s} - b}{\lambda(\delta^2\hat{s} + \mu)}
  \sim \frac{1}{\lambda \delta T \sqrt{2}} \frac{k}{\sqrt{\log k}}.
\end{equation}
Using these properties in~\eqref{eq:merton dens} yields~\eqref{eq:merton dens expl}.
For the density of the underlying itself we thus obtain
\begin{equation}\label{eq:dens S}
  f_{S_T}(k) =f_{X_T}(\log k)/k
  = k^{-(\sqrt{2}/\delta) \sqrt{\log \log k}} \cdot e^{O(\log k)}.
\end{equation}
The marginal law of the underlying has ``almost'' a power law tail,
due to the very slow increase of $\sqrt{\log \log k}$,
but it is still asymptotically lighter than that of any power law.
The influence of the model parameters seems a bit surprising here:
Neither the jump size nor the Poisson intensity appear in the main
factor in~\eqref{eq:dens S}, but only the jump size variance.

To obtain the explicit refined volatility expansion~\eqref{eq:merton V expl},
note that,
by~\eqref{eq:merton call}, \eqref{eq:merton sp}, and~\eqref{eq:merton sp2},
we can refine~\eqref{eq:merton call 1st order} to
\begin{equation}\label{eq:L}
  L = \frac{\sqrt{2}}{\delta} k \sqrt{\log k} - \frac{\mu+\delta^2}{\delta^2}k +
  O\left(\frac{k \log\log k}{\sqrt{\log k}} \right).
\end{equation}
Since
\[
  G_-(k,u) = \frac{\sqrt{2}}{2}k u^{-1/2} - \frac{\sqrt{2}}{8}k^2 u^{-3/2}
  +O(k^3 u^{-5/2}),
\]
using~\eqref{eq:L} in~\eqref{eq:merton V} yields~\eqref{eq:merton V expl}.

We proceed to the proofs of the first equalities in
Theorem~\ref{mjd_call_expansion} and Corollary~\ref{mjd_impvola_expansion}.
The proof of Theorem~\ref{mjd_density}  is very similar to that of
Theorem~\ref{mjd_call_expansion}, using the Fourier representation
of the density, and is omitted (see~\cite{Zr13}).
(Alternatively, the density asymptotics could be deduced from
its series representation, (4.12) in~\cite{CoTa04}, by the Laplace method.)
Formula~\eqref{eq:merton V} in Corollary~\ref{mjd_impvola_expansion}
is a special case of Corollary~7.1 in~\cite{GaLe11}; the error
term follows from~\eqref{eq:merton call 1st order}.
Our Theorem~\ref{mjd_call_expansion} is then useful for
approximating~$L$ in~\eqref{eq:merton V} numerically, or symbolically
to obtain the explicit expansion~\eqref{eq:merton V expl}; see above.
Note that our refined call approximation yields the second order
term (and higher order terms as well) in~\eqref{eq:merton V expl},
which cannot be be deduced from~\eqref{eq:merton call 1st order}.

It remains to prove~\eqref{eq:merton call s}.
Again, we appeal to the representation~\eqref{call_mellin}, where $\eta>1$,
and shift the integration contour to the saddle point~$\hat{s}$.
For the central range, we let $s=\hat{s} + it$, where $|t|<k^{-\alpha}$
with $\alpha \in (\tfrac13,\tfrac12)$.

\begin{lemma}\label{cumulants_mjd}
The cumulant generating function $m(s,T)=\log M(s,T)$ satisfies for, $k \to \infty$, 
\begin{align*}
 m(\hat{s},T) &= \frac{k}{\delta\sqrt{2\log  k}}\left(1+O\left(\frac{\log\log k}{\log k}\right)\right), \\
 m'(\hat{s},T) &= k, \\
 m''(\hat{s},T) &= \delta k \sqrt{2\log k}\left(1+O\left(\frac{1}{\sqrt{\log k}}\right)\right), \\
 m'''(\hat{s}+it,T) &= 2k\log k\left(1+O\left(\frac{1}{\log k}\right)\right),
\end{align*}
where $|t|<k^{-\alpha}$.
\end{lemma}
\begin{proof}
  As for the Kou model, these expansions follow by a straightforward
  computation from the explicit mgf. For the second equation,
  note that we are using the exact saddle point, and not an approximation
  as we did in the Kou model.
\end{proof}
Since the saddle point tends to infinity, the rational function
$1/(s(s-1))$ locally tends to zero:
\[
  \frac{1}{(\hat{s}+it)(\hat{s}-1+it)}
  = \frac{\delta^2}{2\log k}\left(1+O\left(\frac{1}{\sqrt{\log  k}}\right)\right).
\]
(We have used~\eqref{eq:merton sp} here.)
Using this and Lemma~\ref{cumulants_mjd}, we see that the integral
over the central range has the asymptotics~\eqref{eq:merton call s},
after handling the Gaussian integral as in~\eqref{eq:gauss}.
To complete the proof, we need to provide a tail estimate, to the
effect that the integral outside of $s=\hat{s} + it$ with $|t|<k^{-\alpha}$
is negligible. Again, it suffices to treat the upper portion, by symmetry.
\begin{lemma}\label{mjd_tails}
Let $ \tfrac13 < \alpha < \tfrac12$. Then we have
\[
\frac{e^k}{2\pi i}\int_{\hat{s} +ik^{-\alpha}}^{\hat{s}+i\infty} \frac{e^{-ks}M(s,T)}{s(s-1)}ds \ll \frac{e^{k(1-\hat{s})}M(\hat{s},T)e^{-\delta k^{1-2\alpha}/(2\sqrt{2\log k})}}{\log k}.
\]
\end{lemma}
\begin{proof}
  We first bound $M(s,T)$:
  \begin{align*}
|M(s,T)| &= \exp\bigg(\Re\bigg(T\Big(\frac{\sigma^2}{2}s^2 + bs + \lambda\big(e^{\delta^2s^2/2 + \mu s}-1\big)\Big)\bigg)\bigg) \\
&=  \exp\bigg(T\Big(\frac{\sigma^2}{2}(\hat{s}^2-t^2) + b\hat{s} + \lambda\big(\cos(\delta^2 t \hat{s} +\mu t)e^{\delta^2(\hat{s}^2-t^2)/2 + \mu \hat{s}}-1\big)\Big)\bigg)\\
&\leq  \exp\bigg(T\Big(\frac{\sigma^2}{2}(\hat{s}^2-t^2) + b\hat{s} + \lambda\big(e^{\delta^2\hat{s}^2 /2 + \mu \hat{s}}e^{-\delta^2k^{-2\alpha}/2}-1\big)\Big)\bigg).
\end{align*}
  Using
  \[
    e^{-\delta^2k^{-2\alpha}/2} = 1 -\delta^2k^{-2\alpha}/2 + O(k^{-4\alpha})
  \]
  and~\eqref{eq:merton sp2}, we get
  \[
    M(s,T) \ll e^{-T\sigma^2t^2 /2} M(\hat{s},T)e^{-\delta k^{1-2\alpha}/(2\sqrt{2\log k})}.
  \]
  Since $|s(s-1)| \gg \log k$, this estimate implies
   \begin{align*}
\frac{e^k}{2\pi i}\int_{\hat{s} +ik^{-\alpha}}^{\hat{s}+i\infty} \frac{e^{-ks}M(s,T)}{s(s-1)}ds &\ll \frac{e^{k(1-\hat{s})}M(\hat{s},T)e^{-\delta k^{1-2\alpha}/(2\sqrt{2\log k})}}{\log k}\int_{k^{-\alpha}}^{\infty} e^{-\sigma^2t^2 T/2} dt \\
&\ll \frac{e^{k(1-\hat{s})}M(\hat{s},T)e^{-\delta k^{1-2\alpha}/(2\sqrt{2\log k})}}{\log k}.
\end{align*}
\end{proof}

\bibliographystyle{siam}
\bibliography{../gerhold}

\end{document}